\documentclass[11pt]{article}
\usepackage{hyperref}
\usepackage{amssymb}
\usepackage{amsmath}
\usepackage{amsthm}
\usepackage{amstext}
\usepackage{fullpage}
\usepackage{epsfig}
\usepackage{mytex}

\newtheorem{theorem}{Theorem}
\newtheorem{proposition}[theorem]{Proposition}
\newtheorem{definition}[theorem]{Definition}
\newtheorem{corollary}[theorem]{Corollary}
\newtheorem{lemma}[theorem]{Lemma}

\begin{document}

\newcommand{\ignore}[1]{}

\title{Robust Simulations and Significant Separations}
\date{}
\author{Lance Fortnow\thanks{Supported in part by NSF grants CCF-0829754 and DMS-0652521.}\\Northwestern
University \and Rahul Santhanam\\University of Edinburgh}

\maketitle

\begin{abstract}
We define and study a new notion of ``robust simulations''
between complexity classes which is intermediate between the
traditional notions of infinitely-often and almost-everywhere,
as well as a corresponding notion of ``significant
separations''. A language $L$ has a robust simulation in a
complexity class $\C$ if there is a language in $\C$ which
agrees with $L$ on arbitrarily large polynomial stretches of
input lengths. There is a significant separation of $L$ from
$\C$ if there is no robust simulation of $L \in \C$.

The new notion of simulation is a cleaner and more natural
notion of simulation than the infinitely-often notion. We show
that various implications in complexity theory such as the
collapse of $\PH$ if $\NP = \P$ and the Karp-Lipton theorem
have analogues for robust simulations. We then use these
results to prove that most known separations in complexity
theory, such as hierarchy theorems, fixed polynomial circuit
lower bounds, time-space tradeoffs, and the recent theorem of
Williams, can be strengthened to significant
separations, though in each case, an almost everywhere
separation is unknown.

Proving our results requires several new ideas, including a
completely different proof of the hierarchy theorem for
non-deterministic polynomial time than the ones previously
known.

\end{abstract}

\thispagestyle{empty}

\newpage

\setcounter{page}{1}

\section{Introduction}
What does the statement ``$\P\neq\NP$'' really tell us? All is
says is that for any polynomial-time algorithm $A$, $A$ fails
to solve $SAT$ on an infinite number of inputs. These
hard-to-solve inputs could be exponentially (or much worse) far
from each other. Thus even a proof of $\P \neq \NP$ could leave
open the possibility that $SAT$ or any other $\NP$-complete
problem is still solvable on all inputs encountered in
practice. This is unsatisfactory if we consider that one of the
main motivations of proving lower bounds is to understand the
limitations of algorithms.

Another important motivation for proving lower bounds is that
hardness is {\it algorithmically useful} in the context of
cryptography or derandomization. Again, if the hardness only
holds for inputs or input lengths that are very far apart, this
usefulness is called into question. For this reason, theorists
developed a notion of almost-everywhere (a.e.) separations, and
a corresponding notion of infinitely-often (i.o.) simulations.
 A language $L$ is in $\io \C$ for a complexity
class $\C$ if there is some $A\in \C$ such that $A$ and $L$
agree on infinitely many input lengths. A class $\D$ is almost
everywhere not in $\C$ if for some language $L$ in $\D$,
$L\not\in\io \C$, that is any $\C$-algorithm fails to solve $L$
on all but a finite number of input lengths. As an example of
applying these notions, Impagliazzo and Wigderson
\cite{Impagliazzo-Wigderson97} show that if $\E \not \subseteq
\SIZE(2^{o(n)})$ then $\BPP$ is in $\io\P$, and that if $\E
\not \subseteq \io \SIZE(2^{o(n)})$, then $\BPP = \P$.

However, the infinitely often  notion has its own issues.
Ideally, we would like a notion of simulation to capture
``easiness'' in some non-trivial sense. Unfortunately, many
problems that we consider hard have trivial infinitely often
simulations. For example, consider any $\NP$-hard problem on
graphs or square matrices. The natural representation of inputs
for such problems yields non-trivial instances only for input
lengths that are perfect squares. In such a case, the problem
has a trivial infinitely often  simulation on the set of all
input lengths which are not perfect squares. On the other hand,
the problem could be ``padded'' so that it remains non-trivial
on input lengths which are not perfect squares. It's rather
unsatisfactory to have a notion of simulation which is so
sensitive to the choice of input representation.

Not unrelated to this point is that analogues of many classical
complexity results fail to hold in the infinitely often
setting. For example, we do not know if $SAT \in \io\P$ implies
that the entire Polynomial Hierarchy has simulations
infinitely-often in polynomial time. We also don't know if in
general a complete language for a class being easy
infinitely-often implies that the entire class is easy
infinitely-often. This is true for $SAT$ and $\NP$ because
$SAT$ is paddable and downward self-reducible, but for general
complexity classes the implication is unclear. Given that even
these basic analogues are not known, it's not surprising that
more involved results such as the Karp-Lipton theorem
\cite{Karp-Lipton82} and the theorem of Impagliazzo, Kabanets
and Wigderson \cite{Impagliazzo-Kabanets-Wigderson02} don't
have known infinitely often  analogues either.

In an ideal world, we would like all our algorithms to work on
all input lengths, and all our separations to be
almost-everywhere separations. While algorithm design does
typically focus on algorithms that work on all input lengths,
many of the complexity separations we know do not work in the
almost everywhere setting. Separations proved using
combinatorial or algebraic methods, such as Hastad's lower
bound for Parity \cite{Hastad86} or Razborov's monotone circuit
lower bound for Clique \cite{Razborov85} tend to be almost
everywhere (in an appropriate input representation). However,
such techniques typically have intrinsic limitations, as they
run into the natural proofs barrier \cite{Razborov-Rudich97}.
Many of the lower bounds proved recently have come from the use
of indirect diagonalization. A contrary upper bound is assumed
and this assumption is used together with various other ideas
to derive a contradiction to a hierarchy theorem. These newer
results include hierarchy theorems \cite{Barak02,
Fortnow-Santhanam04, vanMelkebeek-Pervyshev06}, time-space
tradeoffs \cite{Fortnow00,
Fortnow-Lipton-vanMelkebeek-Viglas05}, and circuit lower bounds
\cite{Buhrman-Fortnow-Thierauf98, Vinodchandran05, Santhanam07,
Williams10, Williams10a}. Unfortunately, {\it none} of these
results give almost everywhere separations, and so the question
immediately arises what we can say quantitatively about these
separations, in terms of the frequency with which they hold.

To address all these issues, we describe a new notion of
``robust simulation'' and a corresponding notion of
``significant separation''. A language $L$ is in $\ro \C$
(robustly-often in $\C$) if there is a language $A$ in $\C$
such that for every $k$ there are infinitely many $m$ and such
that $A$ and $L$ agree on all inputs between $m$ and $m^k$. A
class $\D$ has a significant separation from $\C$ if there is
some $L$ in $\D$ such that $L\not\in\ro \C$. This implies that
for each $L' \in \C$, there is a constant $k$ such that for
each $m$, $L$ and $L'$ differ on at least one input length
between $m$ and $m^k$. Intuitively, this means that if the
separation holds at some input length, there is another input
length at most polynomially larger at which the separation also
holds, i.e., the hardness is not too ``sparsely'' distributed.

Our definition of robust simulations extends the notion of
uniform hardness of Downey and Fortnow~\cite{Downey-Fortnow}. A
set $A$ is uniformly hard in the sense of Downey and Fortnow if
$A \not \in \ro \P$.

The notion of robust simulation is just slightly stronger than
the notion of infinitely often  simulation, and correspondingly
the notion of significant separation is slightly weaker than
that of almost everywhere separations. By making this tradeoff,
however, we show that we can often achieve the best of both
worlds.

We give robustly often analogues of many classical complexity
results, where infinitely often analogues remain open,
including
\begin{itemize}

\item $\NP \subseteq \ro \P$ implies $\PH \subseteq \ro \P$
\item $\NP \subseteq \ro \SIZE(poly)$ implies $\PH
    \subseteq \ro \P$
\item $\NEXP \subseteq \ro \SIZE(poly)$ implies $\NEXP
    \subseteq \ro \MA$

\end{itemize}

We then use these robustly often  analogues together with other
ideas to give several significant separations where almost
everywhereseparations remain open, including

\begin{itemize}
\item $\NTIME(n^r) \not \subseteq \ro \NTIME(n^s)$, when $r
    > s\geq 1$
\item For each constant $k$, $\Sigma_2 \P \not \subseteq
    \ro \SIZE(n^k)$
\item $SAT \not \subseteq \ro \DTISP(n^{\alpha},
    \polylog(n))$ when $\alpha < \sqrt{2}$
\item $\NEXP \not \subseteq \ro \ACC$
\end{itemize}

The robustly often  notion gives us a cleaner and more powerful
theory than the infinitely often notion.
\subsection{Intuition and Techniques}

To illustrate the advantages of the robustly often  notion over
the infinitely often notion, let's look at a simple example:
trying to show that if $\SAT$ is easy, then all of $\NP$ is
easy. Let $L \in \NTIME(n^k)$ be any $\NP$ language, where $k >
0$ is a constant. $\SAT \in \io \P$ doesn't immediately imply
$L \in \io \P$, as the range of the reduction from $L$ to
$\SAT$ might only intersect input lengths where the
polynomial-time algorithm for $\SAT$ is incorrect. In this
case, the problem can be fixed by padding the reduced instance
to polynomially many different input lengths and using downward
self-reducibility to check YES answers for any of these inputs.
However, this fix depends on specific properties of $\SAT$.

Showing that $\SAT \in \ro \P$ implies $L \in \io \P$ is an
easier, more generic argument. Define a robust set of natural
numbers to be any set $S$ such that for each $k > 0$ there is
an $m$ for which $S$ contains all numbers between $m$ and $m^k$
for some $m$. $\SAT \in \ro \P$ means that there is a robust
set $S$ on which the simulation works. Now the reduction from
$L$ to $\SAT$ creates instances of length $n^k \polylog(n)$,
and it's not hard to see that this automatically implies that
composing the reduction with the algorithm for $\SAT$ gives an
algorithm for $L$ which works on some robust subset $S'$ of
$S$. This implies $L \in \ro \P$. We call a robust subset of a
set a {\it robust refinement}. Many of our arguments will
involve defining a series of robust refinements of a robust set
such that the desired simulation holds on the final refinement
in the series, thus implying that the simulation goes through
in the robustly often  setting.

Thus far, using robustly often  seems easier but hasn't given
us any additional power. The situation changes if we consider
implications where the assumption is used {\it two or more}
times. An example is the proof that $\NP \subseteq \P$ implies
$\PH \subseteq \P$ - this is an inductive proof where the
assumption is used several times. Trying to carry an infinitely
often simulation through fails miserably in this setting
because two infinitely often  simulations do not compose - they
might work on two completely different infinite sets of input
lengths.

Now two robustly often  simulations do not in general compose
either. It is not in general true that for complexity classes
$\B, \C$ and $\D$, $\B \subseteq \ro \C$ and $C \subseteq \ro
\D$, then $\B \subseteq \ro \D$. However, we {\it can} get
these two robustly often  simulations to compose when they are
{\it both} consequences of a single robustly often  assumption.
The robustly often  assumption gives us some robust set to work
with. If we're careful we can define a single robust refinement
of this set on which $\B \subseteq \C$ holds and so too does
$\C \subseteq \D$, which implies $\B \subseteq \D$ holds on
this refinement as well.

This is an idea that we will use again and again in our proofs.
However, in order to use this idea, we need to be careful with
the steps in our proof, as there are only some kinds of
implications for which the idea is useful. For example, it
works well with fixed polynomial-time reductions or
translations with fixed polynomial advice, but not with
exponential padding. More importantly, the idea only works when
all the steps in the proof follow from a {\it single}
assumption, so we need to re-formulate proofs so that they
conform to this pattern. In some cases, eg. the proofs of
Theorem \ref{IKWro} and Theorem \ref{CktLowerBound}, the
re-formulation is non-trivial and leads to proofs that are
quite a bit more involved than the originals
\cite{Impagliazzo-Kabanets-Wigderson02, Kannan82}.

In the case of significant hierarchies for non-deterministic
time, i.e., hierarchies where the lower bound is against robust
simulations, the known techniques break down entirely. The
traditional argument is a ``chaining argument'' \cite{Cook72,
Seiferas-Fischer-Meyer78, Zak83} which uses a chain of
exponentially many input lengths and cannot possibly give a
hierarchy against robustly often  simulations. Here, we come up
with a novel idea of chaining using witnesses to get a
significant hierarchy for polynomial time.

Our most technically involved result is that the recent
breakthrough lower bound of Williams \cite{Williams10,
Williams10a} can be strengthened to a significant separation.
The proof of this result uses almost all of the techniques we
develop, including the sophisticated use of robust refinements
involved in proving Theorem \ref{IKWro}, and a variant of the
significant hierarchy for non-deterministic polynomial time.

Implicit in our paper is a certain proof system for proving
complexity class separations, such that any separation proved
in this system automatically yields an infinitely often
separation. It's an interesting open problem to make this
system explicit, and to study its power and its limitations
more formally.

\section{Preliminaries}
\label{prelims}

\subsection{Complexity Classes, Promise Problems and Advice}

We assume a basic familiarity with complexity classes such as
$\P$, $\RP$, $\BPP$, $\NP$, $\MA$, $\AM$, $\Sigma^{p}_2$,$\PP$
and their exponential-time versions.The Complexity
Zoo\footnote{\tt http://qwiki.caltech.edu/wiki/ComplexityZoo}
is an excellent resource for basic definitions and statements
of results.

Given a complexity class $\C$, $\coC$ is the class of languages
$L$ such that $\bar{L} \in \C$. Given a function $s: \mathbb{N}
\rightarrow \mathbb{N}$, $\SIZE(s)$ is the class of Boolean
functions $f = \{f_n\}$ such that for each $n$, $f_n$ has
Boolean circuits of size $O(s(n))$. Given a language $L$ and an
integer $n$, $L_{n} = L \cap \{0,1\}^{n}$. Given a class $\C$,
$\io\C$ is the class of languages $L$ for which there is a
language $L' \in \C$ such that $L_{n} = L'_{n}$ for infinitely
many length $n$.

In order to deal with promise classes in a general way, we take
as fundamental the notion of a complexity measure. A complexity
measure $\CTIME$ is a mapping which assigns to each pair
$(M,x)$, where $M$ is a time-bounded machine (here a time
function $t_M(x)$ is implicit) and $x$ an input, one of three
values ``0'' (accept), ``1'' (reject) and ``?'' (failure of
$\CTIME$ promise). We distinguish between {\it syntactic} and
{\it semantic} complexity measures. Syntactic measures have as
their range $\{0,1\}$ while semantic measures may map some
machine-input pairs to ``?''. The complexity measures $\DTIME$
and $\NTIME$ are syntactic (each halting deterministic or
non-deterministic machine either accepts or rejects on each
input), while complexity measures such as $\BPTIME$ and
$\MATIME$ are semantic (a probabilistic machine may accept on
an input with probability 1/2, thus failing the bounded-error
promise). For syntactic measures, any halting machine defines a
language, while for semantic measures, only a subset of halting
machines define languages.

A promise problem is a pair $(Y,N)$, where $Y,N \subseteq
\{0,1\}^{*}$ and $Y \cap N = \emptyset$. We say that a promise
problem $(Y,N)$ belongs to a class $\CTIME(t)$ if there is a
machine $M$ halting in time $t$ on all inputs of length $n$
such that $M$ fulfils the $\CTIME$ promise on inputs in $Y \cup
N$, accepting on inputs in $Y$ and rejecting on inputs in $N$.

A language $L$ is in $\CTIME(t)/a$ if there is a machine $M$
halting in time $t(\cdot)$ taking an auxiliary {\it advice}
string of length $a(\cdot)$ such that for each $n$, there is
some advice string $b_n, |b_n|=a(n)$ such that $M$ fulfils the
$\CTIME$ promise for each input $x$ with advice string $b_n$
and accepts $x$ iff $x \in L$.

For syntactic classes, a lower bound with advice or for the
promise version of the class translates to a lower bound for
the class itself.

\begin{definition}
\label{robustness} Let $S$ be a subset of positive integers.
$S$ is robust if for each positive integer $k$, there is a
positive integer $m\geq 2$ such that $n \in S$ for all $m \leq
n \leq m^k$.
\end{definition}

Note that any robust set is infinite. We now define what it
means to simulate a language in a complexity class on a subset
of the positive integers.

\begin{definition}
\label{setsimulations} Let $L$ be a language, $\C$ a complexity
class, and $S$ a subset of the positive integers. We say $L \in
\C$ on $S$ if there is a language $L' \in \C$ such that $L_n =
L'_n$ for any $n \in S$.
\end{definition}

Using the terminology of Definition \ref{setsimulations}, $L
\in \io\C$ for a language $L$ and complexity class $\C$ if
there is some infinite set $S \subseteq \mathbb{N}$ such that
$L \in \C$ on $S$. We now define our main notion of
robustly-often simulations.

\begin{definition}
\label{robustsim} Given a language $L$ and complexity class
$\C$, $L \in \ro\C$ if there is a robust $S$ such that $L \in
C$ on $S$. In such a case, we say that there is a
robustly-often ($\ro$) simulation of $L$ in $\C$. We extend
this notion to complexity classes in the obvious way - given
complexity classes $\B$ and $\C$, $\B \subseteq \ro\C$ if there
for each language $L \in \B$, $L \in\ro\C$. If $\B \not
\subseteq \ro\C$, we say that there is a significant separation
of $\B$ from $\C$.
\end{definition}

Clearly $\B \subseteq \ro\C$ implies $\B \subseteq \io\C$.
Conversely, $\B \not \subseteq \io\C$ gives a very strong
separation of $\B$ and $\C$, i.e., an almost-everywhere
separation, while a significant separation is somewhat weaker
but still much more significant than simply a separation of
$\B$ and $\C$. Intuitively, a significant separation means that
input lengths witnessing the separation are at most
polynomially far apart.

We now define a sequence of {\it canonical refinements} for any
given set $S$, which will play an important part in many of our
proofs.

\begin{definition}
\label{canonical} Let $S$ be a robust set. The canonical
refinement $S_d$ of $S$ at level $d$ is defined as follows for
any integer $d > 0$: $m \in S_d$ iff $m \in S$ and $n \in S$
for all $m \leq n \leq m^d$.
\end{definition}

It is easy to see $S_d$ is robust if $S$ is robust and that
$S_d\subseteq S_{d'}$ for $d\geq d'$.

\section{Robust Simulations}


For any $\NP$-complete language $L$ the language
\[L'=\{x10^i\in L\ |\ \mbox{$|x|+1+i$ is even}\}\] remains
$\NP$-complete but sits in $\io\P$. In contrast if any
$\NP$-complete set under honest m-reductions sits in $\ro\P$
then $\NP\subseteq\ro\P$.

\begin{lemma}
\label{ComplRef} Let $L$ and $L'$ be languages such that $L'$
reduces to $L$ via a polynomial-time honest m-reduction. Let
$\C$ be a complexity class closed under poly-time m-reductions.
If there is a robust $S$ such that $L \in \C$ on $S$, then
there is a robust refinement $S'$ of $S$ such that $L' \in \C$
on $S'$.
\end{lemma}

Proof in appendix.

The proof ideas of Lemma~\ref{ComplRef} can be used to show
that robustly often  analogues of various useful implications
hold. We omit the proofs of these propositions, since they
follow from the definition of robustly often  in much the same
way as Lemma \ref{ComplRef}.

The first analogue essentially says that we can take a complete
language to be representative of a complexity class, even in
the context of robustly often  simulations. It is an immediate
consequence of Lemma \ref{ComplRef}.

\begin{proposition}
\label{SATro} If $\SAT \in \ro \P$, then $\NP \subseteq \ro
\P$.
\end{proposition}

In fact, if $\SAT \in \P$ on $S$ for some robust set $S$ and $L
\in \NTIME(n^d)$ for some integer $d$, the proof of Proposition
\ref{SATro} gives that $L \in \P$ on $S_{d+1}$, i.e., on the
canonical refinement of $S$ at level $d+1$.

The next proposition says that translation arguments using a
fixed polynomial amount of advice carry through in the robustly
often setting, for any ``reasonable'' complexity measure where
``reasonable'' means that the measure is closed under efficient
deterministic transductions \cite{vanMelkebeek-Pervyshev06}.
All complexity measures considered in this paper are reasonable
in this sense.

\begin{proposition}
\label{PolyPadding} Let $\BTIME$ and $\CTIME$ be any complexity
measures closed under efficient deterministic transductions.
Let $g$ and $h$ be time-constructibe functions, and $p$ a
polynomial. If $\BTIME(g(n)) \subseteq \ro \CTIME(h(n))$, then
$\BTIME(g(p(n))) \subseteq \ro \CTIME(h(p(n)))$.
\end{proposition}

As a consequence of Proposition \ref{PolyPadding}, we get for
example that if $\NTIME(n) \subseteq \ro \P$, then $\NP
\subseteq \ro \P$. There are contexts where exponential padding
is used in complexity theory, eg., in the proof that $\NP = \P$
implies $\NEXP = \EXP$. This result doesn't seem to translate
to the robustly often  setting, however a weaker version does.

\begin{proposition}
\label{NEXPro} If $\NP \subseteq \ro \P$, then $\NEXP \subseteq
\io \EXP$.
\end{proposition}

The proposition below says that simulations of a syntactic
class in another class can be translated to a simulation with
fixed polynomial advice, even in the robustly often  setting.

\begin{proposition}
\label{PolyAdvice} Let $\BTIME$ be a syntactic complexity
measure and $\CTIME$ a complexity measure, such that both
$\BTIME$ and $\CTIME$ are closed under efficient deterministic
transductions. Let $f$ and $g$ be time-constructrible measures
and $p$ a polynomial. If $\BTIME(f(n)) \subseteq \ro
\CTIME(g(n))$, then $\BTIME(f(n))/p(n) \subseteq \ro
\CTIME(g(n+p(n)))/p(n)$.
\end{proposition}

\begin{theorem}
\label{PolyHier} If $\NP \subseteq \ro\P$, then $\PH \subseteq
\ro\P$
\end{theorem}

Proof in the appendix.

\begin{theorem}
\label{PolyHierNonunif} If $\NP \subseteq \ro\SIZE(poly)$, then
$\PH \subseteq \ro\SIZE(poly)$
\end{theorem}

\begin{theorem}
\label{PolyHierBPP} If $\NP \subseteq \ro\BPP$, then $\PH
\subseteq \ro\BPP$.
\end{theorem}

We omit the proof of Theorem \ref{PolyHierNonunif}, which
closely resembles the proof of Theorem \ref{PolyHier}.

Next we show a robust analogue of the Karp-Lipton theorem
\cite{Karp-Lipton82}. We formulate a stronger statement which
will be useful when we show significant fixed-polynomial
circuit size lower bounds for $\Sigma^{p}_{2}$. The proof is
simpler than for some of the other analogues, since the
assumption is only used once during the proof.

\begin{lemma}
\label{KLRobust} If there is a constant $k$ and a robust set
$S$ such that $\SAT \in \SIZE(n^k)$ on $S$, then there is a
robust refinement $S'$ of $S$ such that $\Pi_2SAT \in
\Sigma_2-\TIME(n^{k+1+o(1)})$.
\end{lemma}

Proof in the appendix.

The following is an immediate corollary.

\begin{corollary}
\label{KLcorollary} If $\NP \subseteq \ro\SIZE(poly)$, then
$\Sigma^{p}_{2} \subseteq \ro\Pi^{p}_{2}$.
\end{corollary}

Analogously, it is easy to show an robustly often version of a
theorem credited to Nisan by Babai, Fortnow and Lund
\cite{Babai-Fortnow-Lund91}. We need the following lemma.

\begin{lemma}
\label{KLMNro} Let $L$ be a language complete for $\E$ under
linear-time reductions. If there is a robust set $S$ and a
constant $k \geq 1$ such that $L \in \SIZE(n^k)$ on $S$, then
there is a robust refinement $S'$ of $S$ such that $L \in
\MATIME(n^{2k})$ on $S'$.
\end{lemma}

The analogue of Nisan's result follows as a corollary.

\begin{corollary}
\label{Nisanro} If $\EXP \subseteq \ro\SIZE(poly)$, then $\EXP
\subseteq \ro\MA$.
\end{corollary}

The following can be shown using the easy witness method of
Kabanets \cite{Kabanets01, Impagliazzo-Kabanets-Wigderson02}
and known results on pseudo-random generators
\cite{Nisan-Wigderson94, Klivans-vanMelkebeek99}.

\begin{lemma}
\label{EasyWitness} Let $R$ be any robust set and $k > 1$ be
any constant. Then there is a robust refinement $R'$ of $R$
such that either $\NE \subseteq \DTIME(2^{n^{16k^4}})$ on $R$
or $\MATIME(n^{4k^2}) \subseteq \NE/O(n)$ on $R'$.
\end{lemma}

We next derive an robustly often  analogue of the main theorem
of Impagliazzo, Kabanets and Wigderson
\cite{Impagliazzo-Kabanets-Wigderson02}.

\begin{theorem}
\label{IKWro} If $\NEXP \subseteq \ro \SIZE(poly)$, then $\NEXP
\subseteq \ro \MA$.
\end{theorem}

Proof in the appendix.

\section{Significant Separations}

\subsection{Hierarchies}

The proofs of the hierarchies for deterministic time and space
actually give almost-everywhere separations and therefore
significant separations.

For nondeterministic time the situation is quite different.
Cook~\cite{Cook72} showed that
$\NTIME(n^r)\subsetneq\NTIME(n^s)$ for any reals $r<s$.
Seiferas, Fischer and Meyer~\cite{Seiferas-Fischer-Meyer78}
generalize this result to show that
$\NTIME(t_1(n))\subsetneq\NTIME(t_2(n))$ for
$t_1(n+1)=o(t_2(n))$. Zak~\cite{Zak83} gives a simpler proof of
the same result. All these proofs require building an
exponential (or worse) chain of equalities to get a
contradiction. Their proofs do not give almost everywhere
separations or significant separations. No relativizable proof
can give an $\io$ hierarchy as Buhrman, Fortnow and Santhanam
give a relativized world that $\NEXP\subseteq\io\NP$.

In this section we give relativizing proofs that
$\NEXP\not\subseteq\ro\NP$ and that
$\NTIME(n^r)\not\subseteq\ro \NTIME(n^s)$ for $r>s\geq 1$. The
latter proof requires a new proof of the traditional
nondeterministic time hierarchy.

%
%

\begin{theorem}
\label{NexpHier} $\NEXP\not\subseteq\ro\NP$.
\end{theorem}

Proof in appendix. Using similar ideas, we can get a separation
against sub-polynomial advice.

\begin{theorem}
\label{NexpHierAdv} $\NEXP \not \subseteq \ro \NP/n^{o(1)}$
\end{theorem}

In the purely uniform setting, we now show a stronger version
of Theorem \ref{NexpHier} in the form of a significant
hierarchy for non-deterministic polynomial time.

\begin{theorem}
\label{ntime} If $t_1$ and $t_2$ are time-constructible
functions such that
\begin{itemize}
\item $t_1(n)=o(t_2(n))$, and
\item $n\leq t_1(n)\leq n^c$ for some constant $c$
\end{itemize}
then $\NTIME(t_2(n))\not\subseteq \ro\NTIME(t_1(n))$.
\end{theorem}

\begin{corollary}
For any reals $1\leq r<s$, $\NTIME(n^s)\not\subseteq
\ro\NTIME(n^r)$.
\end{corollary}

\begin{proof}[Proof of Theorem~\ref{ntime}]
Let $M_1,M_2,\ldots$ be an enumeration of multitape
nondeterministic machines that run in time $t_1(n)$.

Define a nondeterministic Turing machine $M$ that on input
$1^i01^m0w$ does as follows:
\begin{itemize}
\item If $|w|<t_1(i+m+2)$ accept if both $M_i(1^i01^m0w0)$
    and $M_i(1^i01^m0w1)$ accepts.
\item If $|w|\geq t_1(i+m+2)$ accept if $M_i(1^i01^m0)$
    rejects on the path specified by the bits of $w$.
\end{itemize}
Since we can universally simulate $t(n)$-time nondeterministic
multitape Turing machines on an $O(t(n))$-time 2-tape
nondeterministic Turing machine,
$L(M)\in\NTIME(O(t_1(n+1)))\subseteq\NTIME(t_2(n))$. Note
$(n+1)^c=O(n^c)$ for any $c$.

Suppose $\NTIME(t_2(n))\subseteq\ro\NTIME(t_1(n))$. Pick a $c$
such that $t_1(n)\ll n^c$. By the definition of $\ro$ there is
some $n_0$ and a language $L\in\NTIME(t_1(n))$ such that
$L(M)=L$ on all inputs of length between $n_0$ and $n_0^c$. Fix
$i$ such that $L=L(M_i)$. Then $z\in L(M_i)\Leftrightarrow z\in
L(M)$ for all $z=1^i01^{n_0}0w$ for $w\leq t_1(i+n_0+2)$.

By induction we have $M_i(1^i01^{n_0}0)$ accepts if
$M_i(1^i01^{n_0}0w)$ accepts for all $w\leq t_1(i+n_0+2)$. So
$M_i(1^i01^{n_0}0)$ accepts if and only $M_i(1^i01^{n_0}0)$
rejects on every computation path, contradicting the definition
of nondeterministic time.
\end{proof}

\subsection{Circuit Lower Bounds}

Kannan \cite{Kannan82} showed how to diagonalize in
$\Sigma^{p}_{4}$ against circuit size $n^k$ almost everywhere,
for any constant $k$. He combined this result with the
Karp-Lipton theorem to obtain fixed polynomial circuit size
lower bounds in $\Sigma^{p}_2$. However, the lower bound for
$\Sigma^{p}_2$ no longer holds almost everywhere, as a
consequence of the proof strategy used by Kannan. It has been
open since then to derive an almost everywhere separation in
this context. We manage to obtain a significant separation. We
need a quantitative form of Kannan's original diagonalization
result. Kannan's paper contains a slightly weaker form of this
result, but his proof does yield the result below.

\begin{theorem} \cite{Kannan82}
\label{Sigma4LowerBound} For any constant $k \geq 1$,
$\Sigma_4-\TIME(n^{3k}) \not \subseteq \io\SIZE(n^k)$
\end{theorem}

\begin{theorem}
\label{CktLowerBound} For each integer $k > 1$, $\Sigma^{p}_2
\not\subseteq\ro\SIZE(n^k)$.
\end{theorem}

Proof in the appendix.

Using similar ideas we can get robustly often analogues of the
lower bound of Cai and Sengupta \cite{Cai01} for $\StP$ and
Vinodchandran \cite{Vinodchandran05} for $\PP$:

\begin{theorem}
\label{S2PLowerBound} For any $k > 0$, $\StP \not \subseteq \ro
\SIZE(n^k)$.
\end{theorem}

\begin{theorem}
\label{PPLowerBound} For any $k > 0$, $\PP \not \subseteq \ro
\SIZE(n^k)$.
\end{theorem}

The following strong separation can be shown by direct
diagonalization.

\begin{proposition}
\label{DiagCirc} For any integer $k \geq 1$,
$\DTIME(2^{n^{2k}}) \not \subseteq \io\SIZE(n^k)$.
\end{proposition}

We use Proposition \ref{DiagCirc} to give an analogue of the
result of Burhman, Fortnow and Thierauf
\cite{Buhrman-Fortnow-Thierauf98} that $\MAEXP \not \subseteq
\SIZE(\poly)$ in the robustly often  setting. We give a circuit
lower bound for promise problems in $\MAEXP$ rather than
languages. Intuitively, the fact that $\MATIME$ is a semantic
measure causes problems in extending the proof of Buhrman,
Fortnow and Thierauf to the robustly often  setting; however,
these problems can be circumvented if we consider promise
problems.

\begin{theorem}
\label{MAEXP} $Promise-\MAEXP\not\subseteq\ro\SIZE(poly)$.
\end{theorem}

\begin{proof}
Let $Q$ be a promise problem complete for $\Promise-\MAE$ under
linear-time reductions. Assume, for the purpose of
contradiction, that $\Promise-\MAEXP \subseteq \ro\SIZE(poly)$.
Then there is a robust set $S$ and an integer $k > 0$ such that
$Q \in \SIZE(n^k)$ on $S$. Let $L$ be a language complete for
$\E$ under linear-time reductions. Since there is a linear-time
reduction from $L$ to $Q$, it follows that there is a robust
refinement $S'$ of $S$ such that $L \in \SIZE(n^k)$ on $S$.
Using Lemma, there is a robust refinement $S''$ of $S'$ such
that $L \in \MATIME(n^{2k})$ on $S''$. Using a simple padding
trick, there is a robust refinement $S'''$ of $S''$ such that
for each language $L' \in \DTIME(2^{n^{2k}})$, $L' \in
\MATIME(n^{4k^2})$ on $S'''$. Using completeness of $Q$ again,
there is a robust refinement $S''''$ of $S'''$ such that $L'
\in SIZE(n^k)$ on $S''''$ for any $L' \in \DTIME(2^{n^{2k}})$,
which contradicts Proposition \ref{DiagCirc}.
\end{proof}

Our most technically involved result is that the recent lower
bound of Williams~\cite{Williams10a} that
$\NEXP\not\subseteq\ACC$ extends to the robustly often setting.
His proof uses the nondeterministic time hierarchy and the
proof of Impagliazzo, Kabanets and
Wigderson~\cite{Impagliazzo-Kabanets-Wigderson02}, neither of
which may hold in the infinitely-often setting. So to get a
robustly-often result we require variants of our
Theorems~\ref{ntime} and~\ref{IKWro}. To save space, we will
focus on the new ingredients, and abstract out what we need
from Williams' paper.

We first need the following simultaneous resource-bounded
complexity class.

\begin{definition}
\label{NTIGU} $\NTIMEGUESS(T(n), g(n))$ is the class of
languages accepted by NTMs running in time $O(T(n))$ and using
at most $O(g(n))$ non-deterministic bits.
\end{definition}

We have the following variant of Theorem~\ref{ntime}, which has
a similar proof.

\begin{lemma}
\label{ntimeguess} For any constant $k$, $\NTIME(2^n) \not
\subseteq \ro \NTIMEGUESS(2^{n}/n, n^k)$.
\end{lemma}

We also need the following robustly often  analogue of a
theorem of Williams \cite{Williams10}, which uses the proof
idea of Theorem \ref{IKWro}. The problem $\StSAT$ is complete
for $\NEXP$ under polynomial-time m-reductions.

\begin{lemma}
\label{univro} If $\NE \subseteq \ro \ACC$ on $S$ for some
robust set $S$, then there is a constant $c$ and a refinement
$S'$ of $S$ such that $\StSAT$ has succinct satisfying
assignments that are $\ACC$ circuits of size $n^c$ on $S'$.
\end{lemma}

\begin{proof}
The proof of Theorem \ref{IKWro} gives that if $\NE \subseteq
\ACC$ on $S$, then there is a constant $d$ and a robust
refinement $R$ of $S$ such that $\StSAT$ has succinct
satisfying assignments that are circuits of size $n^d$ on $R$.
Since $\P \subseteq \ACC$ on $S$ and using Proposition
\ref{PolyAdvice}, we get that there is a constant $c$ and a
robust refinement $S'$ of $R$ such that $\StSAT$ has succinct
satisfying assignments that are $\ACC$ circuits of size $n^c$
on $S'$.
\end{proof}

Now we are ready to prove the robustly often  analogue of
Williams' main result \cite{Williams10a}.

\begin{theorem}
\label{Williamsro} $\NEXP \not \subseteq \ro \ACC$.
\end{theorem}

\begin{proofsketch}
Assume, to the contrary, that $\StSAT \in \ACC$ on $R$ for some
robust $R$. By completeness of $\StSAT$, it follows that there
is a robust refinement $S$ of $R$ and a constant $k'
> 1$ such that $\NE$ has $\ACC$ circuits of size $n^{k'}$. Let
$L \in \NTIME(2^n)$ but not in $\ro \NTIMEGUESS(2^n/n, n^k)$,
where $k$ will be chosen large enough as a function of $k'$.
Existence of $L$ is guaranteed by Lemma \ref{ntimeguess}. We
will show $L \in \ro \NTIMEGUESS(2^n/n, n^k)$ and obtain a
contradiction.

The proof of Theorem 3.2 in Williams' paper gives an algorithm
for determinining if $x \in L$. The algorithm
non-deterministically guesses and verifies a "small" (size
$n^{O(c)}$) $\ACC$ circuit which is equivalent to the $\StSAT$
instance to which $x$ reduces, within time $2^n/\omega(n)$ by
using Williams' new algorithm for $\ACC$-SAT together with the
assumption that $\NEXP$ and hence $\P$ in $\ACC$ on $S$. This
guess-and-verification procedure works correctly on some robust
refinement of $S$.  Then, the algorithm uses the existence
guarantee of Lemma \ref{univro} to guess and verify a succint
witness, again using Williams' algorithm for $\ACC$-SAT. This
further guess-and-verification procedure works correctly on
some further robust refinement $S''$ of $S$. In total, the
algorithm uses at most $n^{dk'}$ non-deterministic bits for
some constant $d$, runs in time at most $2^n/n$ and decides $L$
correctly on $S''$. By choosing $k
> dk'$, we get the desired contradiction.
\end{proofsketch}

Williams' work still leaves open whether $\NEXP \subseteq
\SIZE(poly)$. Using the same ideas as in the proof of Theorem
\ref{Williamsro}, we can show that an algorithm for
$\CircuitSAT$ that improves slightly on brute force search
robustly often would suffice to get such a separation.

\begin{theorem}
\label{AlgFromRoLb} If for each polynomial $p$, $\CircuitSAT$
can be solved in time $2^{n - \omega(\log(n))}$ robustly often
on instances where the circuit size is at most $p(n)$, then
$\NEXP \not \subseteq \SIZE(poly)$.
\end{theorem}

\subsection{Time-Space Tradeoffs}

\begin{proposition}
\label{nondetcohier} Let $t$ and $T$ be time-constructible
functions such that $t = o(T)$. Then $\NTIME(T) \not \subseteq
\io\coNTIME(t)$.
\end{proposition}

\begin{theorem}
\label{TimeSpTrade} Let $\alpha < \sqrt{2}$ be any constant.
$\SAT \not \in \ro\DTISP(n^{\alpha}, \polylog(n))$.
\end{theorem}

Proof in the appendix. Similar ideas can be used to show the
best known time-space tradeoffs for $\SAT$ - the key is that
the proofs of these tradeoffs proceed by indirect
diagonalization, using the contrary assumption and polynomial
padding a constant number of times and deriving a contradiction
to a strong hierarchy theorem.

\section{Conclusion and Open Problems}
Our paper makes significant progress in remedying one of the
main issues with proofs using indirect diagonalization - the
fact that they don't give almost everywhere separations. We
have shown that most interesting results proved using this
technique can be strengthened at least to give significant
separations.

There are still many separations which we don't know how to
make significant: eg., the main results of
\cite{Buhrman-Fortnow-Santhanam09, Santhanam07,
Buhrman-Fortnow-Thierauf98}. The reasons why we're unable to
get significant separations are different in different cases:
the inability to formulate the proof as following from a single
assumption about robustly often  simulations in the first case,
the use of superpolynomial amount of padding in the second, and
the fact that complete languages for $\MAEXP$ are unknown in
the last. It's not clear whether there are inherent limitations
to getting significant separations in these cases.

Other variants of robustly often  might be interesting to
study, such as a ``linear'' variant where the simulation is
only required to hold on arbitrarily large linear stretches of
inputs rather than polynomial stretches. Separations against
this notion of simulation would be stronger than significant
separations.

\bibliographystyle{alpha}
\bibliography{papers_Lance}

\newpage

\section*{Appendix}

Here we give proofs omitted from the body of the paper.

\begin{proof}[Proof of Lemma~\ref{ComplRef}]
Let $f$ be a polynomial-time honest m-reduction from $L'$ to
$L$. By assumption, there is a robust $S$ such that $L \in \C$
on $S$. Let $K \in C$ be a language such that $L_n = K_n$ for
all $n \in S$. We define a robust refinement $S'$ of $S$ as
follows: $n' \in S'$ iff $n' \in S$ and for all $x$ of length
$n'$, $|f(x)| \in S$. By definition, $S'$ is a refinement of
$S$; we show that it is robust and that $L' \in C$ on $S'$.

First, we show robustness of $S'$. Robustness of $S$ is
equivalent to saying that for each positive integer $k$, there
is a positive integer $m(k)$ such that $n \in S$ for all
$m(k)^{1/k} \leq n \leq m(k)^k$. Since $f$ is an honest
reduction, there is an integer $c > 1$ such that for all $x$,
$|x|^{1/c} \leq |f(x)| \leq |x|^c$. We show that for each
positive integer $k$, there is a positive integer $m'(k)$ such
that $n' \in S'$ for all $m'(k)^{1/k} \leq n' \leq m'(k)^k$.
Simply choose $m'(k) = m(ck)$. We have that for any $n'$ such
that $m(ck)^{1/k} \leq n' \leq m(ck)^k$, for all $x$ of length
$n'$, $|f(x)|$ is between $(n')^{1/c}$ and $(n')^c$ by
assumption on $f$. Hence for all $x$ of length $n'$, $|f(x)|$
is between $m(ck)^{1/ck}$ and $m(ck)^{ck}$, which implies $f(x)
\in S$ for all such $x$, and hence $n' \in S'$.

Next we show $L' \in \C$ on $S'$. Define a language $K'$ as
follows: $x \in K'$ iff $f(x) \in K$. Since $\C$ is closed
under poly-time m-reductions, $K' \in \C$. Now, on each input
length $n' \in S'$, for any $x$ of length $n'$, $|f(x)| \in S$,
and hence $L(x) = K(x)$. This implies that $L'(x) = K'(x)$ for
all such $x$, and hence $L'_{n'} = K'_{n'}$, which finishes the
proof.
\end{proof}

\begin{proof}[Proof of Theorem~\ref{PolyHier}]
We show that for any positive integer $k$, $\Sigma^p_{k}
\subseteq \ro\P$. The proof is by induction on $k$. We will
need to formulate our inductive hypothesis carefully.

By assumption, $\SAT \in \ro\DTIME(n^c)$ for some integer $c >
0$. Let $S$ be a robust set such that $\SAT \in \DTIME(n^c)$ on
$S$. Using the same idea as in the proof of Lemma
\ref{ComplRef}, we have that for any $L \in \NTIME(n^d)$, $L
\in \DTIME(n^{cd})$ on $S_d$, where $S_d$ is the canonical
refinement defined in Definition \ref{canonical}. Moreover, the
conclusion holds even if the assumption is merely that $L \in
\NTIME(n^d)$ on $S_d$.

Now we formulate our inductive hypothesis $H_k$. The hypothesis
$H_k$ is that $\Sigma_{k}SAT \in \DTIME(n^{c^{k+o(1)}})$ on
$S_{c^{k+1}}$. For $k=1$, the hypothesis holds because by
assumption $\SAT \in DTIME(n^c)$ on $S$ and hence on $S_{c^2}$
since $S_{c^2}$ is a refinement of $S$. If we prove the
inductive hypothesis, the proof of the theorem will be complete
using Lemma \ref{ComplRef} and the fact that $\Sigma_{k}-SAT$
is complete for $\Sigma^{p}_{k}$ under polynomial-time honest
m-reductions.

Assume that $H_k$ holds - we will establish $H_{k+1}$. Since
$H_k$ holds, we have that $\Sigma_{k}SAT \in
\DTIME(n^{c^{k+o(1)}})$ on $S_{c^{k+1}}$. Now $\Sigma_{k+1}SAT
\in \NTIME(n^{1+o(1)})^{\Sigma_{k}SAT}$. Moreover, by using
padding, we can assume that that all the queries of the oracle
machine are of length $n^{1+o(1)}$. Now, using the same idea as
in the proof of Lemma \ref{ComplRef},  we have that
$\Sigma_{k+1}SAT \in \NTIME(n^{c^{k+o(1)}})$ on $S_{c^{k+2}}$.
Using again the assumption that $\SAT \in \DTIME(n^c)$ on
$S$,we have for each language $L \in \NTIME(n^{c^{k+o(1)}})$ on
$S_{c^{k+1}}$, $L \in \DTIME(n^{c^{k+1+o(1)}})$ on
$S_{c^{k+1}}$ and hence on $S_{c^{k+2}}$ since $S_{c^{k+2}}$ is
a refinement of $S_{c^{k+1}}$. Thus we have that
$\Sigma_{k+1}SAT \in \DTIME(n^{c^{k+1+o(1)}})$ on
$S_{c^{k+1}}$, which establishes $H_{k+1}$ and completes the
proof of the theorem.
\end{proof}

\begin{proof}[Proof of Lemma~\ref{KLRobust}]

We follow the usual proof of the Karp-Lipton theorem. Let
$\{C_n\}$ be a sequence of circuits such that $C_n$ solves
$\SAT$ correctly for $n \in S$ and $|C_n| = O(n^k)$. Using
self-reducibility and paddability of $\SAT$, we can define a
sequence $\{C'_n\}$ of circuits such that $C'_n$ outputs a
satisfying assignment for all satisfiable formulae of length $n
\in S$, and $|C'_n| = O(n^{k+1})$ for all $n$.

Now let $\phi$ be an instance of $\Pi_2SAT$: $\phi \in
\Pi_2SAT$ iff $\forall \vec{x} \exists \vec{y}
\phi(\vec{x},\vec{y})$. By $\NP$-completeness and paddability
of $\SAT$, there is a polynomial-time computable function $f$
such that $\phi \in \Pi_2SAT$ iff $\forall \vec{x}
f(<\phi,\vec{x}>) \in SAT$, and $|z| \leq |f(z)| \leq
|z|^{1+o(1)}$. Consider the following $\Sigma_2$ algorithm for
$\phi$: it existentially guesses a circuit $C'$ of size
$O(n^{k+1})$ and universally verifies for all $\vec{x}$ that
$\phi(\vec{x}, C'(f(<\phi, \vec{x}>)))$ holds. This algorithm
decides $\Pi_2SAT$ correctly on all $n \in S_2$, where $S_2$ is
the canonical refinement of $S$ at level $2$. The time taken by
the algorithm is $O(n^{k+1+o(1)})$, which gives the desired
conclusion.
\end{proof}

\begin{proof}[Proof of Theorem~\ref{IKWro}]
Let $k > 1$ be a constant and $L$ be a complete set for $\NE$
under linear-time reductions for which there is a robust set
$S$ such that $L \in \SIZE(n^k)$ on $S$. This implies that
there is a robust refinement $S'$ of $S$ such that $\NE
\subseteq \SIZE(n^k)$ on $S'$, using Lemma \ref{ComplRef}.
Using Lemma \ref{KLMNro}, there is a robust refinement $S''$ of
$S'$ such that $\E \subseteq \MATIME(n^{2k})$ on $S''$. Using
Proposition \ref{PolyPadding}, we have that there is a robust
refinement $S'''$ of $S''$ such that $\DTIME(2^{n^{2k}})
\subseteq \MATIME(n^{4k^2})$ on $S'''$. Now set $R=S'''$ in
Lemma \ref{EasyWitness}.

There are two cases. Case 1 is that we have $\NE \subseteq
\DTIME(2^{n^{16k^4}})$ on $S''$. In this case, by using padding
and applying Proposition \ref{PolyPadding}, we have that $\NE
\subseteq \DTIME(2^{n^{16k^4}}) \subseteq \MA$ on
$S''_{16k^4+1}$. This gives that $\NE \subseteq \ro \MA$, which
implies $\NEXP \subseteq \ro \MA$ by applying Proposition
\ref{PolyPadding} again, and we have the desired conclusion.

Case 2 is that $\MATIME(n^{4k^2}) \subseteq \NE/O(n)$ on some
robust refinement $S''''$ of $S'''$. We will derive a
contradiction in this case. Since $S''''$ is a refinement of
$S'''$, we have that $\DTIME(2^{n^{2k}}) \subseteq
\MATIME(n^{4k^2}) \subseteq \NE/O(n)$ on $S''''$. Since $S''''$
is a refinement of $S'$, we have that $\NE \subseteq
\SIZE(n^k)$ on $S''''$. Applying Proposition \ref{PolyAdvice},
we have that there is a robust refinement $S'''''$ of $S''''$
such that $\NE/O(n) \subseteq \SIZE(n^k)$ on $S'''''$. Thus
$\DTIME(2^{n^{2k}}) \subseteq \ro \SIZE(n^k)$, but this is in
contradiction to the fact that $\DTIME(2^{n^{2k}}) \not
\subseteq \io \SIZE(n^k)$ by direct diagonalization.
\end{proof}

\begin{proof}[Proof of Theorem~\ref{NexpHier}]
Assume, to the contrary, that there is an infinitely-often
robust simulation of $\NEXP$ by $\NP$. Let $L$ be a language
which is paddable and complete for $\NE$ under linear-time
m-reductions. By assumption, there is a robust set $S$ and an
integer $k$ such that $L \in \NTIME(n^k)$ on $S$.

Let $K$ be a set in $\DTIME(2^{n^{2k}})$ but not in
$\io\NTIME(n^k)$ - such a set can be constructed by direct
diagonalization. Consider the following padded version $K'$ of
$K$: $y \in K'$ iff $y = x01^{|x|^{2k} - |x| - 1}$ for $x \in
K$. Since $K \in \DTIME(2^{n^{2k}})$, we have that $K' \in \E
\subseteq \NE$. Now using the assumption on $L$ and the proof
idea of Lemma \ref{ComplRef}, we have that $K' \in \NTIME(n^k)$
on $S_2$, where $S_2$ is the canonical refinement of $S$ at
level $2$. Since there is an m-reduction from $K$ to $K'$
running in time $O(n^{2k})$ which only blows up the instance
length, we can use the proof idea of Lemma \ref{ComplRef} again
to show that $K \in \NTIME(n^{2k^2})$ on $S_{4k}$. Now, since
$K \in \NE$ and $L$ is complete for $\NE$, we can use the proof
of idea of Lemma \ref{ComplRef} a third time to show that $K
\in \NTIME(n^k)$ on $S_{4k}$. But since $S_{4k}$ is infinite,
this implies $K \in \NTIME(n^k)$ on an infinite set of input
lengths, which contradicts our assumption on $K$.

\end{proof}

\begin{proof}[Proof of Theorem~\ref{CktLowerBound}]
Assume, contrarily, that $\Sigma^{p}_2 \subseteq
\ro\SIZE(n^k)$. Let $L$ be a set that is complete for
$\Sigma_2-\TIME(n^{4k^4})$ under linear-time reductions, and
$S$ be a robust set such that $L \in \SIZE(n^k)$ on $S$.  Hence
also $\Sigma_{2}SAT \in \SIZE(n^k)$ on $S$. This implies $\SAT
\in \SIZE(n^k)$ on $S$, and by Lemma \ref{KLRobust}, there is a
robust refinement $S'$ of $S$ such that $\Pi_2SAT \in
\Sigma_2-\TIME(n^{k+o(1)})$ on $S'$. Now, by using the proof
idea of Theorem \ref{PolyHier}, we have that there is a robust
refinement $S''$ of $S'$ such that $\Sigma_4SAT \in
\Sigma_2-\TIME(n^{k^3+o(1)})$ on $S''$. Using the fact that
$\Sigma_4SAT$ is complete for $\Sigma_4 \P$ and the proof idea
of Lemma \ref{ComplRef}, we have that there is a robust
refinement $S'''$ of $S''$ such that $\Sigma_4-\TIME(n^{3k})
\subseteq \Sigma_2-\TIME(n^{3k^4 + o(1)})$ on $S'''$. Now, by
completeness of $L$ for $\Sigma_2-\TIME(n^{4k^4})$, we have
that there is a robust refinement $S''''$ of $S'''$ such that
$\Sigma_2-\TIME(n^{3k^4+o(1)}) \subseteq \SIZE(n^k)$ on
$S''''$. This implies $\Sigma_4-\TIME(n^{3k}) \subseteq
\SIZE(n^k)$ on $S''''$, which contradicts Theorem
\ref{Sigma4LowerBound} since $S''''$ is an infinite set.
\end{proof}

\begin{proof}[Proof of Theorem~\ref{TimeSpTrade}]
Assume, to the contrary, that $\SAT \in \ro\DTISP(n^{\alpha},
\polylog(n))$ on $S$ for some robust set $S$ and constant
$\alpha < \sqrt{2}$. Since $\DTISP(n^{\alpha}, \polylog(n))
\subseteq \Pi_2-\TIME(n^{\alpha/2 + o(1)})$
\cite{Fortnow-Lipton-vanMelkebeek-Viglas05} and using that
$\SAT$ is complete for $\NTIME(n \polylog(n))$ under
quasilinear-time length-increasing reductions, we have that
there is a robust refinement $S'$ of $S$ such that $\NTIME(n)
\subseteq \DTISP(n^{\alpha}, \polylog(n)) \subseteq
\Pi_2-\TIME(n^{\alpha/2+o(1)})$ on $S'$.  By padding, there is
a robust refinement $S''$ of $S'$ such that $\NTIME(n^2)
\subseteq \Pi_2-\TIME(n^{\alpha + o(1)})$ on $S''$. By using
the assumption that $\NTIME(n) \subseteq \DTIME(n^{alpha})$ on
$S'$ again to eliminate the existential quantifier in the
$\Pi_2$ simulation, we have that there is a robust refinement
$S'''$ of $S''$ such that $\NTIME(n^2) \subseteq
\coNTIME(n^{\alpha^2 + o(1)})$ on $S'''$. But this is a
contradiction to Proposition \ref{nondetcohier}, since $S'''$
is infinite.
\end{proof}

\end{document}